\documentclass[10pt]{elsarticle}
\usepackage{amssymb,amsmath}
\usepackage{enumitem}
\usepackage{bbm}
\usepackage{upgreek}
\usepackage{mathtools}
\usepackage[margin=1in,papersize={8.5in,11in}]{geometry}
\usepackage[utf8]{inputenc}
\usepackage[english]{babel}
\usepackage{amsthm}
\usepackage{float}
\usepackage{amsfonts,mathrsfs}
\usepackage{bm}
\usepackage{makecell}
\usepackage{tikz}
\usetikzlibrary{matrix}

\usepackage{epstopdf}
\usepackage{subcaption}
\usepackage{tablefootnote}
\graphicspath{ {images/} }
\usepackage{comment}

\newtheorem{theorem}{Theorem}
\newtheorem{corollary}{Corollary}[theorem]
\newtheorem{lemma}[theorem]{Lemma}
\newtheorem{prop}{Proposition}
\theoremstyle{definition}

\def\R{\mathbb{R}}
\def\H{\mathcal{H}}

\def\W{\mathcal{W}}
\def\A{\mathcal{A}}

\def\D{\mathcal{D}} 
 
\def\L{\mathcal{L}} 
\def\M{\mathcal{M}}

\def\P{\mathcal{P}}

\def\X{\mathcal{X}}

\def\Ran{\text{Ran}} 
\def\Ker{\text{Ker}}

\def\Q{\mathcal{Q}}

\def\K{\mathcal{K}} 
\def\I{\mathcal{I}}

\def\Pi{\mathbf{P}}

\journal{arXiv}

\begin{document}
\begin{frontmatter}
\title{A simple hypocoercivity analysis for the effective Mori-Zwanzig equation}
\author[ucm]{Yuanran Zhu}
\address[ucm]{Department of Applied Mathematics, University of California Merced\\ Merced (CA) 95343}
\cortext[correspondingAuthor]{Corresponding author}
\ead{yzhu56@ucmerced.edu}
\begin{abstract}
We provide a simple hypocoercivity analysis for the effective Mori-Zwanzig equation governing the time evolution of noise-averaged observables in a stochastic dynamical system. Under the hypocoercivity framework mainly developed by Dolbeault, Mouhot and Schmeiser and further extended by Grothaus and Stilgenbauer, we prove that under the same conditions which lead to the geometric ergodicity of the Markov semigroup $e^{t\K}$, the Mori-Zwanzig orthogonal semigroup $e^{t\Q\K\Q}$ is also geometrically ergodic, provided that $\P=\I-\Q$ is a finite-rank, orthogonal projection operator in a certain Hilbert space. The result is applied to the widely used Mori-type effective Mori-Zwanzig equations in the coarse-grained modeling of molecular dynamics and leads to exponentially decaying estimates for the memory kernel and the fluctuation force. 
\end{abstract}
\end{frontmatter}
\section{Introduction}
Coarse-grained modeling has become an important theme in the numerical simulations of complex systems. This is also a rather difficult topic because in essence it requires projecting the high-dimensional dynamics in a low-dimensional phase space, and the closure problem in such a dimension-reduction process arises since the unresolved dynamics and its interaction with the reaction coordinate is generally unknown. 
In recent years, the Mori-Zwanzig (MZ) formalism has become popular in the physics and applied mathematics communities to tackle the coarse-graining problem. This framework was introduced in nonequilibrium statistical mechanics \cite{Mori,zwanzig1961memory} to explain the non-Markovian properties of certain observables in a high-dimensional Hamiltonian system. Later, its extension to stochastic systems has been considered by different researchers \cite{morita1980contraction,espanol1995hydrodynamics,hudson2018coarse,zhu2020hypoellipticity}. The most attracting feature of the MZ theory is that it allows us to systematically derive exact evolution equations, now known as the generalized Langevin equations (GLEs), for any quantities of interest based on the microscopic equations of motion. 
Such GLEs can be used as the ansatz for coarse-grained models which found applications in molecular dynamics 
\cite{Li2015,VenturiBook,Yoshimoto2013,
espanol1995statistical,espanol1995hydrodynamics,zhu2021effective,zhu2021generalized}, 
fluid mechanics \cite{parish2017non,parish2017dynamic}, 
and, more generally, systems described by nonlinear 
partial differential equations (PDEs)
\cite{venturi2014convolutionless,chorin2000optimal,
stinis2007higher,stinis2004stochastic,lin2021data,lu2017data}. 
In this paper, we are mainly concerned with the Mori-Zwanzig theory for stochastic systems. Since the MZ equation we are going to use is for noise-averaged observables of the stochastic system (see details in Section 2), it will be called the effective Mori-Zwanzig (EMZ) equation to make us distinguish from the regular MZ equation for Hamiltonian systems. Although being widely used in coarse-grained modeling, a mathematically rigorous study of the EMZ equation is still scarce in the community. This has to do with the well-known difficulty in analyzing the orthogonal propagator of the EMZ equation, which can be represented by a strongly continuous semigroup $e^{t\Q\K\Q}$, where $\Q$ is a projection operator and $\K$ is the Kolmogorov backward operator of the SDE. Previous work in this direction aimed at obtaining existence conditions for the orthogonal dynamics under Mori's projection \cite{givon2005existence}. Until recently, the study of the EMZ equation has been push forward by the work of Zhu and Venturi \cite{zhu2020hypoellipticity}. The main finding in \cite{zhu2020hypoellipticity} is that $e^{t\Q\K\Q}$ corresponds to a linear integro-differential, kinetic equation, hence one can use the hypoelliticity method mainly developed by
H\'erau \& Nier \cite{herau2004isotropic}, Helffer \& 
Nier \cite{nier2005hypoelliptic} and Eckmann \& Hairer \cite{eckmann1999non, eckmann2000non,eckmann2003spectral} to obtain the spectrum properties of operator $\Q\K\Q$. In particular, it is shown that \cite{zhu2020hypoellipticity} the spectrum of $\Q\K\Q$ is discrete and confined in a cusp in $\mathbb{C}$, just like the spectrum of $\K$. As a direct consequence of this, the ergodicity of the Markov semigroup $e^{t\K}$ implies the ergodicity of the EMZ orthogonal semigroup $e^{t\Q\K\Q}$, provided that $\P=\I-\Q$ is a finite-rank, symmetric projection operator. This result leads to exponential decaying prior estimates for the EMZ memory kernel and fluctuation force term, and the later ones are often used as the prerequisite for the reduced-order modeling of large-scale molecular systems \cite{Li2015,li2017,wang2019implicit}.

The success of the hypoellipticity method inspires us to develop a hypocoercivity theory for analyzing the dynamical behavior of the EMZ equation. The hypocoercivity method was mainly developed by Villani \cite{villani2009hypocoercivity}, Dolbeault, Mouhot and Schmeiser \cite{dolbeault2015hypocoercivity} and many other researchers \cite{grothaus2014hypocoercivity,grothaus2016hilbert,grothaus2020hypocoercivity,leimkuhler2020hypocoercivity,leimkuhler2017ergodic} to study the kinetic equations in a pure functional analysis framework. For detailed explorations in this regard, we refer to the above papers and the reference therein. The general strategy of this method is to derive the coercivity estimate in terms of an equivalent Hilbert/Sobolev norm, thereby obtain the geometric ergodicity of the corresponding semigroup. In this work, we will show that the elegant hypocoercivity framework introduced by Dolbeault, Mouhot and Schmeiser (DMS) \cite{dolbeault2015hypocoercivity} and further extended by Grothaus and Stilgenbauer (GS) \cite{grothaus2014hypocoercivity,grothaus2016hilbert} fits well in the study of the EMZ equation. Specifically, by restricting the estimate of $e^{t\Q\K\Q}$ in a special invariant Hilbert subspace, we provide a simple proof of its geometric ergodicity under the {\em same} hypocoercivity conditions that lead to the geometric ergodicity of the Markov semigroup $e^{t\K}$. This allows us to get a hypocoercivity version of the aforementioned exponential decaying estimates for the EMZ memory kernel and the fluctuation force. The abstract analysis result is then applied to the coarse-graining problem of the Langevin dynamics. 

This paper is organized as follows. Section \ref{sec:EMZE} reviews the derivation of the effective Mori-Zwanzig (EMZ) equation for stochastic dynamical systems and its application to the coarse-graining problem. In Section \ref{sec:ana}, we first review the DMS-GS hypocoercivity framework \cite{dolbeault2015hypocoercivity,grothaus2014hypocoercivity,grothaus2016hilbert} for the analysis of semigroup $e^{t\K}$. Then we prove that similar results hold for the EMZ semigroup $e^{t\Q\K\Q}$. In Section \ref{sec:app}, the abstract analysis result is applied to the Langevin dynamics to obtain the exponentially decaying estimates for the EMZ memory kernel and the fluctuation force. A comparison between the hypoellipticity method and the hypocoercivity method is also provided. The main findings of this paper are summarized in Section \ref{sec:summary}. We also propose several interesting open problems on the analysis of the EMZ equations.
%

\section{EMZ equation and the coarse-graining of stochastic systems}
\label{sec:EMZE}
Consider a $d$-dimensional, time-homogeneous stochastic 
differential equation (SDE) on $\R^d$:
\begin{align}\label{eqn:sde}
dx(t)=F(x(t))+\sigma(x(t))d\W(t), \qquad x(0)=x_0\sim \rho_0(x),
\end{align}
where $F:\R^d\mapsto \R^d$ and 
$\sigma: \R^d\rightarrow \R^{d\times m}$ are 
smooth functions, $\W(t)$ is the
$m$-dimensional Wiener process, and $x_0$ is a 
random initial state characterized in terms of 
a probability density function $\rho_0(x)$. It is well known that the system of SDE \eqref{eqn:sde} induces a $d$-dimensional Markov process on $\R^d$.  This allows us to define a 
composition operator $\M(t,0)$ that pushes forward 
in time the average of the observable $ u(t)= u( x(t))$ 
over the noise, i.e., 
\begin{align}
\mathbb{E}_{\W(t)}[ u( x(t))| x_0]= 
\M(t,0) u(x_0)=e^{t\K} u(x_0).
\label{MarkovSemi}
\end{align}
With a slight abuse of notation, hereafter we use $ u(t)= u( x(t))$ to represent the observable function and its noise average \eqref{MarkovSemi}. For autonomous system \eqref{eqn:sde}, $\M(t,0)$ is a Markovian semigroup
generated by the following Kolmogorov backward operator 
\cite{Risken,kloeden2013numerical} under the It\^o-interpretation of the white noise:
\begin{align}
\K( x_0)&=\sum_{k=1}^dF_k( x_0)\frac{\partial}{\partial x_{0k}}
+\frac{1}{2}\sum_{j=1}^m\sum_{i,k=1}^d\sigma_{ij}( x_0)
\sigma_{kj}( x_0)\frac{\partial}{\partial x_{0i}\partial x_{0k}}.
\label{KI}
\end{align}
With the evolution operator $\M(t,0)$ available, we can now derive the effective Mori-Zwanzig (EMZ) equation for the noise-averaged observable $ u( x(t))$. To this end, we introduce a projection operator $\P$ and 
the complementary projection $\Q=\I-\P$.
By differentiating Dyson's identity \cite{zhu2020hypoellipticity}, 
we obtain the exact evolution equation governing 
the dynamics of the noise-averaged observable 
\eqref{MarkovSemi}\footnote{Here we replaced the orthogonal dynamics propagator $e^{t\Q\K}$ with $e^{t\Q\K\Q}$. Such a replacement is needed for our later research. It is possible because the evolution operator $e^{t\Q\K}$ and $e^{t\Q\K\Q}$ are equivalent within the range of $\Q$ since $\Q$ is a projection operator satisfying $\Q^2=\Q$.}:
\begin{equation}
\frac{\partial}{\partial t}e^{t\K} u(0)
=e^{t\K}\mathcal{PK} u(0)
+e^{t\Q\K\Q}\mathcal{QK} u(0)+\int_0^te^{s\K}\P\K
e^{(t-s)\Q\K\Q}\mathcal{QK} u(0)ds,\label{eqn:EMZ_full}
\end{equation}
where $ u(0)= u(x_0)$. The three terms at the right hand side of \eqref{eqn:EMZ_full} 
are called, respectively, streaming term, fluctuation 
(or noise) term, and memory term. It is often more 
convenient (and tractable) to compute the evolution 
of $ u(t)$  within a closed 
linear space, e.g., the image of the projection operator 
$\P$. Applying the projection operator 
$\P$ to \eqref{eqn:EMZ_full}
yields
\begin{equation} 
\frac{\partial}{\partial t}\mathcal{P}e^{t\K} u(0)
=\mathcal{P}e^{t\K}\mathcal{PK} u(0)
+\int_0^t\P e^{s\K}\mathcal{PK}
e^{(t-s)\Q\K\Q}\mathcal{QK} u(0)ds,\label{eqn:EMZ_projected}
\end{equation}
where the noise term vanishes since $\P\Q=0$. Due to the fact that equation \eqref{eqn:EMZ_full} and its projected form \eqref{eqn:EMZ_projected} only describe the dynamics of the {\em noise-averaged} observables, they were called as the effective Mori-Zwanzig (EMZ) equation for stochastic systems. The EMZ equations
has the same structure as the classical MZ 
equation for deterministic, Hamiltomian systems 
\cite{zhu2019generalized,zhu2018estimation,zhu2018faber},
with the first-order Liouville operator replaced by a Kolmogorov operator $\K$. Depending on the choice of the projection operator $\P$, the projected EMZ equation can yield evolution equations for different statistics for observable $u(x(t))$, including the time auto-correlation function and the statistical moment \cite{zhu2020hypoellipticity,zhu2021effective}. In this paper, we are mainly concerned with Mori-type projection operator $\P$. To this end, we choose a set of observable functions $\{u_i(t)\}_{i=1}^M$ as the quantities of interest. Assuming that $u_i(t)=u_i(x(t))\in H$, where $H$ is a certain Hilbert space, then Mori's projection operator $\P$ can be defined as:
\begin{align}
\label{Mori_P}
\P h:=\sum_{i,j=1}^M G^{-1}_{ij}
\langle u_i(0),h\rangle u_j(0),
\qquad h\in H,
\end{align}
where Gram matrix $G_{ij}=\langle u_i(0),u_j(0)\rangle$
and $u_i(0)=u_i(x_0)$ ($i=1,...,M$) are linearly independent with respect to the inner product of the Hilbert space $H$. It is easy to check that Mori-type projection operator \eqref{Mori_P} is a finite-rank, orthogonal projection operator in $H$. With $\P$ available, we can rewrite 
the EMZ equation \eqref{eqn:EMZ_full} and its 
projected form \eqref{eqn:EMZ_projected}
as: 
\begin{align}
\frac{d u(t)}{dt} &=  \Omega  u(t) +
\int_{0}^{t} K(t-s) u(s)ds+ f(t),\label{gle_full}\\
\frac{d}{dt}\P{ u}(t) &=  \Omega\P { u}(t)+ 
\int_{0}^{t}  K(t-s) \P {u}(s)ds,\label{gle_projected}
\end{align}
where $ u(t) =[u_1(t),\dots,u_M(t)]^T$, the streaming matrix $\Omega_{ij}= \sum_{k=1}^MG^{-1}_{jk}\langle u_{k}(0), \K u_{i}(0)\rangle$. The quantities that involve the orthogonal dynamics $e^{t\Q\K\Q}$ in \eqref{gle_full}-\eqref{gle_projected} are the EMZ memory kernel $ K(t)$ and the fluctuation force $ f(t)$ \cite{zhu2018estimation,zhu2020hypoellipticity}: 
\begin{align}
K_{ij}(t)&=\sum_{k=1}^M G^{-1}_{jk}\langle u_{k}(0), \K e^{t\Q\K\Q}\Q\K u_{i}(0)\rangle,\label{memory_kernel}\\
f_i(t)&=e^{t\Q\K\Q}\Q\K u_i(0)\label{f_force}.
\end{align}
Equation \eqref{gle_full} and \eqref{gle_projected} are also known as the linear generalized Langevin equation (GLE) in statistical mechanics \cite{mori1965continued,zwanzig2001nonequilibrium}. In coarse-grained modeling of molecular systems \cite{li2010coarse,chen2014computation,ma2016derivation}, the observable set of Mori's projection operator $\P$ may be chosen as: $\{u_i(t)\}_{i=1}^{2M}=\{\bar{p}_i(t),\bar{q}_i(t)\}_{i=1}^M$, where $\bar{p}_i(t)=\sum_{j=1}^{N_{ij}}p_j(t)/N_{ij}$ and $\bar{q}_i(t)=\sum_{j=1}^{N_{ij}}q_j(t)/N_{ij}$ are the momentum and position of the $i$-th coarse-grained particle which are obtained by averaging over the total momentum and position of a cluster of small particles. Then the full EMZ equation \eqref{gle_full} yields the formal evolution equation of the coarse-grained particles and the projected equation \eqref{gle_projected} describes the evolution of their time autocorrelation functions \cite{zhu2020hypoellipticity}.

\section{Dolbeault-Mouhot-Schmeiser hypocoercivity analysis}
\label{sec:ana}
As we have seen from \eqref{memory_kernel}-\eqref{f_force}, the dynamical properties of the fluctuation force $f(t)$ and the memory kernel $K(t)$ in EMZ equations \eqref{eqn:EMZ_full}-\eqref{eqn:EMZ_projected} are determined by the orthogonal semigroup $e^{t\Q\K\Q}$. In this section, we perform a simple hypocoercivity analysis to analyze the dynamical behavior of $e^{t\Q\K\Q}$ using the technique mainly developed by Dolbeault, Mouhot and Schmeiser (DMS) \cite{dolbeault2015hypocoercivity} and further extended by Grothaus and Stilgenbauer (GS) \cite{grothaus2014hypocoercivity,grothaus2016hilbert}. Specifically, we show that under the {\em same} conditions, the abstract Hilbert space analysis result obtained in \cite{dolbeault2015hypocoercivity} (Section 1.3), or more directly related result  \cite{grothaus2016hilbert} (Section 2), can be easily extended to get the geometric ergodicity of the semigroup $e^{t\Q\K\Q}$ within a certain Hilbert subspace. Hence the ergodicity of $e^{t\Q\K\Q}$ can be viewed as the direct result of the 
ergodicity of $e^{-t\K}$, provided that the projection $\P=\I-\Q$ is a symmetric and finite-rank operator, i.e. of the Mori-type.
\subsection{Hypocoercivity analysis for $e^{t\K}$}
We first briefly review the abstract hypocoercivity analysis result of DMS-GS. Since the evolution operator $e^{t\K}$ in the EMZ equations \eqref{eqn:EMZ_full}-\eqref{eqn:EMZ_projected} is generated by the Kolomogorov backward operator $\K$, we will follow the setting in GS's work \cite{grothaus2014hypocoercivity,grothaus2016hilbert}, which can be viewed as a dual result of DMS \cite{dolbeault2015hypocoercivity} for the Kolomogorov backward equation. To be noticed that the notation follows that of \cite{grothaus2016hilbert} with slight modifications for our purposes. Moreover, specific discussions on the operator domains will be omitted and we refer readers to \cite{grothaus2016hilbert} for technical details. 

Suppose the closure of operator $\K$ generates a strongly continuous semigroup $e^{t\K}$ in a certain Hilbert space $H$. Further assume that $\K$ can be decomposed as $\K=\D-\L$ where $\L$, $\D$ are, receptively, a skew-symmetric and symmetric operator in $H$. Now we introduce two auxiliary operators $\Uppi$ and $\A$, where $\Uppi$ is the orthogonal projection on the kernel of $\D$
and $\A$ is defined as 
\begin{align}\label{def_A}
   \A:=(\I+(\L\Uppi)^*(\L\Uppi))^{-1}(\L\Uppi)^*. 
\end{align}
Define the Hilbert space inner product and norm by $\langle\cdot,\cdot\rangle$ and $\|\cdot\|$. The analysis of DMS \cite{dolbeault2015hypocoercivity} and  GS \cite{grothaus2014hypocoercivity,grothaus2016hilbert} introduced  the following assumptions for operators $\L,\D,\Uppi,\A$ in a proper domain $D\subset H$:
\begin{itemize}
\item $(H1)$ (Alegebraic relation): $\Uppi\L\Uppi|_{D}=0$.
\item $(H2)$ (Microscopic coercivity): There exists a constant $\Lambda_m>0$ such that 
\begin{equation}
    -\langle\D f,f\rangle\geq\Lambda_m\|(\I-\Uppi)f\|^2,\qquad \forall f\in D.
\end{equation}
\item $(H3)$ (Macroscopic coercivity): Operator $\Uppi\L^2\Uppi$ is essentially self-adjoint on $H$, Moreover, there exists a constant $\Lambda_M>0$ such that 
\begin{align}
    \|\L\Uppi f\|^2\geq\Lambda_M\|\Uppi f\|^2,\qquad \forall f\in D.
\end{align}
\item  $(H4)$ (Boundedness of auxiliary operators): There exists $c_1,c_2>0$ such that 
\begin{align}
    \|\A\D f\|\leq c_1\|(\I-\Uppi)f\|,\qquad 
    \|\A\L(\I-\Uppi)f\|\leq c_2\|(\I-\Uppi)f\|,\qquad \forall f\in D.
\end{align}
\end{itemize}
If assumptions $(H1)-(H4)$ hold, along with the domain condition (D) detailed in \cite{grothaus2016hilbert}, it can be proved \cite{grothaus2014hypocoercivity,grothaus2016hilbert} that the semigroup $e^{t\K}$ is geometrically ergodic in $D=H\cap\text{Ker}(\K)^{\perp}$ with the estimate
\begin{align}\label{e^tK_estimate}
    \|e^{t\K}f_0-\pi_0f_0\|\leq Ce^{-t\lambda}\|f_0-\pi_0f_0\|, \qquad \forall t\geq 0,\quad f_0\in H,
\end{align}
where $f_0=f(0)$ is the initial condition of an observable function $f(t)$, $\pi_0$ is the orthogonal projection operator onto the kernel of operator $\K$, i.e. $\text{Ker}(\K)$, and $\text{Ker}(\K)^{\perp}$ is the orthogonal complement of $\text{Ker}(\K)$ in $H$. The proof of estimate \eqref{e^tK_estimate} is obtained in an abstract manner using the assumptions $(H1)-(H4)$. To this end, we introduce a modified entropy functional $H_{\epsilon}[f]$ defined as 
\begin{align}\label{modified_entropy}
    H_{\epsilon}[g]:=\frac{1}{2}\|g\|^2+\epsilon\langle\A g,g\rangle,\qquad g\in H.
\end{align}
It can be shown that $H_{\epsilon}[f]$ introduces a norm equivalent of $\|\cdot\|$ since 
\begin{align}
   \frac{1-\epsilon}{2}\|g\|^2 \leq H_{\epsilon}[g]\leq \frac{1+\epsilon}{2}\|g\|^2.
\end{align}
Here we used the property that $\A$ is a bounded operator, which can be derived from $(H1)-(H3)$ (\cite{grothaus2014hypocoercivity}, section 2). With the equivalent norm \eqref{modified_entropy}, to obtain \eqref{e^tK_estimate}, it is sufficient to use the Gronwall's lemma \cite{villani2009hypocoercivity} and show that there exists a $\kappa>0$ such that 
\begin{align*}
    -\frac{d}{dt}H_{\epsilon}[g(t)]=D_{\epsilon}[g(t)]\geq \kappa\|g(t)\|^2
\end{align*}
for all $g_0=f_0-\pi_0 f_0\in \text{Ker}(\K)^{\perp}$, $f_0\in H$. Since $\text{Ker}(\K)^{\perp}$ is an invariant subspace of the operator $\K$ and $e^{t\K}$, the last inequality can be obtained from $(H1)-(H4)$ and the definition \eqref{modified_entropy} since 
\begin{equation}\label{Dg_estimate}
\begin{aligned}
    D_{\epsilon}[g]
    &=-\langle\D g, g\rangle+\epsilon\langle\A\L\Uppi g,g\rangle
    +\epsilon\langle\A\L(1-\Uppi) g,g\rangle
    -\epsilon\langle\L\A g,g\rangle-\epsilon\langle\A\D g,g\rangle\\
    &\geq \Lambda_m\|(\I-\Uppi)g\|^2+\epsilon\frac{\Lambda_M}{1+\Lambda_M}\|\Uppi g\|^2-\epsilon(1+c_3)\|(\I-\Uppi)g\|\|g\|\\
    &\geq 
    \left[\Lambda_m-\epsilon(1+c_3)\left(1+\frac{1}{2\delta}\right)\right]\|(\I-\Uppi)g\|^2+\epsilon\left(\frac{\Lambda_M}{1+\Lambda_M}-(1+c_3)\frac{\delta}{2}\right)\|\Uppi g\|^2\\
    &\geq \kappa \|g\|^2,
\end{aligned}
\end{equation}
where $c_3=c_1+c_2$ and we have chosen a small enough $\delta$ and then a small $\epsilon=\epsilon(\delta)$ to make the last inequality hold. The above method to obtain the geometric ergodicity of semigroup $e^{t\K}$ is called the ``hypocoercivity'' analysis because the Kolmogorov operator $\K$ is coercive in an equivalent norm $H_{\epsilon}[\cdot]$. For details of this method, we refer to Villani's monograph \cite{villani2009hypocoercivity} and DMS's modification of Villani's method for linear kinetic equations \cite{dolbeault2015hypocoercivity}. 
\subsection{Hypocoercivity analysis for $e^{t\Q\K\Q}$}
Consider an orthogonal projection operator $\P: H\rightarrow H$ in Hilbert space $H$. Naturally $\Q=\I-\P:H\rightarrow H$ is also an orthogonal projection operator in $H$. Operator $\Q$ induces a useful function space $\Q H$ defined as 
\begin{align}\label{def:QH}
f\in \Q H=\Ran(\Q)\cap H :=\left\{f|\: \exists g\in H,\: s.t.\: f=\Q g\in H\right\}.
\end{align}
Note that for any $f\in \Q H$, using the idempotency $\Q^2=\Q$ of the projection operator we have $\Q f=\Q\Q g=\Q g=f$. This relation will be used frequently throughout the paper. We first discuss some properties of the function space $\Q H$. Orthogonal projection operators $\P$ and $\Q$ introduce a decomposition of $H$ as $H=\Ran(\Q)\oplus\Ker(\Q)=\Ker(\P)\oplus\Ran(\P)$, where $\Ran(\Q)=\Ker(\P)$ is a subspace of $H$. As a subspace, $\Q H$ inherits all properties of the original Hilbert space $H$, as well as the definitions of the inner product and norm. To analyze the evolution operator $e^{t\Q\K\Q}$, we further define a smaller subspace $\Q H/\text{Ker}(\K)=\Q H\cap\text{Ker}(\K)^{\perp}$, which equips the same inner product and norm as of $H$. Then in $\Q H/\text{Ker}(\K)$, we can get the geometric ergodicity of $e^{t\Q\K\Q}$. The result can be summarized as:
\begin{theorem}\label{thm:QKQ_hypocoercive}
Assume that operator $\K=\D-\L$ satisfies $(H1)-(H4)$ and the domain condition $(D)$ in \cite{grothaus2016hilbert}. If $\D$ can be decomposed as $\D=\X^*\X$, where $\X$ and its adjoint $\X^*$ in $H$ are first-order differential operators, and the projection operator $\P: H\rightarrow H$ is a finite-rank, symmetric operator, then for the EMZ orthogonal semigroup $e^{t\Q\K\Q}$, we have exponentially decaying estimate:
\begin{align}\label{eqn:e^{-tQKQ}}
   \|e^{t\Q\K\Q}f_0\|\leq Ce^{-\lambda_{\Q}t}\|f_0\|, \qquad \forall f_0\in \Q H/\text{Ker}(\K),
\end{align}
where the constants $C$ and $\lambda_{\Q}$ are explicitly computable.
\end{theorem}
\begin{proof}
We first note that the closure of operator $\Q\K\Q$ generates a contraction semigroup, denoted as $e^{t\Q\K\Q}$, in Hilbert space $H$. This fact is proved in Theorem 3 of \cite{zhu2020hypoellipticity} using the symmetry of the orthogonal projection operator $\Q$ and Lumer-Philips theorem. In order to get estimate \eqref{eqn:e^{-tQKQ}}, we need to prove that $\Q H\cap\text{Ker}(\K)^{\perp} $ is an invariant subspace of $\Q\K\Q$, i.e. $\Q\K\Q:\Q H\cap\text{Ker}(\K)^{\perp}\rightarrow \Q H\cap\text{Ker}(\K)^{\perp}$.

To prove this claim, we choose an arbitrary $f\in \Ker(\K)\cap\Q H$. Since $\Q$ is a symmetric operator in $H$, $\L$ is skew-symmetric, we have
\begin{align*}
\mathfrak{R}\langle\Q\K\Q f,f\rangle=
\mathfrak{R}\langle\K\Q f,\Q f\rangle
=\mathfrak{R}\langle\D\Q f,\Q f\rangle
=\mathfrak{R}\langle\D f,f\rangle
=\|\X f\|^2=0.
\end{align*}
Here we have used the fact that $\Q f=f$ for $f\in\Q H$. The above equality implies $f\in \Ker(\X)$. Immediately, we obtain $\K^*f=2\D f-\K f=2\X^*\X f-\K f=0 $. This yields $\text{Ker}(\K)\cap \Q H\subset \text{Ker}(\K^*)\cap \Q H$. For $h\in \Q H$ and arbitrary $g\in \Ker(\K)\cap \Q H$, we have 
\begin{align*}
\langle \Q\K\Q h,g \rangle=\langle h, \Q\K^*\Q g\rangle=\langle h, \Q\K^* g\rangle=0.
\end{align*}
Since $\Q\K\Q h\in \Q H$, the above result implies $\Q\K\Q h\in \Q H\cap \Ker(\K)^{\perp}$. We have proved that $\Q\K\Q:\Q H\cap\text{Ker}(\K)^{\perp}\rightarrow \Q H\cap\text{Ker}(\K)^{\perp}$. Within this invariant subspace, we can easily prove \eqref{eqn:e^{-tQKQ}}. Using the exact same modified entropy \eqref{modified_entropy} for $g(t)=e^{t\Q\K\Q}g_0$, where $g_0\in\Q H\cap\text{Ker}(\K)^{\perp}$, we differentiate the functional \eqref{modified_entropy} to obtain
\begin{align}\label{entropy_eqn}
    -\frac{d}{dt}H_{\epsilon}[g(t)]=D_{\epsilon}^{\Q}[g(t)],
\end{align}
where 
\begin{align*}
    D^{\Q}_{\epsilon}[g]:=-\langle \Q\K\Q g,g\rangle+\epsilon\langle\A\Q\K\Q g,g\rangle
    +\epsilon\langle\A g,\Q\K\Q g\rangle.
\end{align*}
Using the fact that $\Q g=g$ and $\P=\I-\Q$, we obtain:
\begin{align*}
    D^{\Q}_{\epsilon}[g]&=-\langle \K g, g\rangle+\epsilon\langle\A\Q\K g,g\rangle
    +\epsilon\langle\A g,\Q\K g\rangle\\
     &=-\langle\D  g, g\rangle+\epsilon\langle\A\K g,g\rangle
    +\epsilon\langle\A g,\K g\rangle-\epsilon\langle\A\P\K g,g\rangle
    -\epsilon\langle\A g,\P\K g\rangle\\
    &=\underbrace{-\langle\D g, g\rangle+\epsilon\langle\A\L\Uppi g,g\rangle
    +\epsilon\langle\A\L(1-\Uppi) g,g\rangle
    -\epsilon\langle\L\A g,g\rangle-\epsilon\langle\A\D g,g\rangle}_{(I)}
    \underbrace{
    -\epsilon\langle\A\P\K g,g\rangle
    -\epsilon\langle\A g,\P\K g\rangle}_{(II)}.
\end{align*}
To be noticed that the first five terms $(I)$ are exactly the same as the one for the semigroup $e^{t\K}$, i.e. the right hand side of \eqref{modified_entropy}. Since $\Q H\cap\text{Ker}(\K)^{\perp}$ is a subspace of $H$ which inherits the Hilbert space inner product and norm, we can directly use estimate \eqref{Dg_estimate} to bound $(I)$ as 
\begin{align*}
(I)\geq \Lambda_m\|(\I-\Uppi)g\|^2+\frac{\epsilon\Lambda_M}{1+\Lambda_M}\|\Uppi g\|^2-\epsilon(1+c_3)\|(\I-\Uppi)g\|\|g\|.
\end{align*}
On the other hand, since the projection operator $\P:H\rightarrow H$ is finite-rank and symmetric, therefore admits the canonical form $\P(\cdot)=\sum_{i=1}^N\langle\cdot,\phi_i\rangle\varphi_i$, $(II)$ can be lower bounded as 
\begin{align*}
(II)&=-\epsilon\sum_{i=1}^N\langle\A\varphi_i,g\rangle\langle\K g,\phi_i\rangle
    -\epsilon\sum_{i=1}^N\langle\A g,\varphi_i\rangle\langle\K g,\phi_i\rangle\\
    &\geq-\epsilon\sum_{i=1}^N\|\A^*g\|\|\varphi_i\|\|\K^*\phi_i\|\|g\|
    -\epsilon\sum_{i=1}^N\|\A g\|\|\varphi_i\|\|\K^*\phi_i\|\|g\|\\
    &\geq -\epsilon c_4\|(\I-\Uppi)g\|\|g\|.
\end{align*}
Here we used Cauchy-Schwartz inequality and the property that for bounded operator $\|\A\|=\|\A^*\|\leq \|\I-\Uppi\|$, where the bound of $\A$ can be obtained using $(H1)-(H3)$ \cite{grothaus2016hilbert,dolbeault2015hypocoercivity}. Adding up $(I)$ and $(II)$ we can get:
\begin{equation}\label{DgQ_estimate}
\begin{aligned}
(I)+(II)&\geq \Lambda_m\|(1-\Uppi)g\|^2+\frac{\epsilon\Lambda_M}{1+\lambda_M}\|\Uppi g\|^2-\epsilon(1+c_3+c_4)\|(1-\Uppi)g\|\|g\|\\
&\geq \left[\Lambda_m-\epsilon(1+c_3+c_4)\left(1+\frac{1}{2\delta}\right)\right]\|(1-\Uppi)g\|^2+\epsilon\left[\frac{\lambda_M}{1+\Lambda_M}-(1+c_3+c_4)\frac{\delta}{2}\right]\|\Uppi g\|^2.
\end{aligned}
\end{equation}
By choosing first a small $\delta$ and then $\epsilon=\epsilon(\delta)$, we can find a positive $\kappa_{\Q}$ such that $D_{\epsilon}^{\Q}[g]\geq \kappa_{\Q}\|g\|^2$. Using equation \eqref{entropy_eqn}, we obtain
\begin{align*}
\frac{d}{dt}H_{\epsilon}[g(t)]\leq \frac{-2\kappa_{\Q}}{1+\epsilon}H_{\epsilon}[g(t)].
\end{align*}
With Gronwall's Lemma and the equivalence of the Hilbert norm $\|\cdot\|$ and the entropy norm $H_{\epsilon}[\cdot]$, we can get semigroup estimate \eqref{eqn:e^{-tQKQ}}. 
\end{proof}
The identification of the invariant subspace $\Q H\cap\text{Ker}(\K)^{\perp}$ is crucial for the hypocoercivity method to work. This is the main technical part of the proof. In fact, we further show that for operator $\Q\K\Q$, the restriction of its orthogonal kernel projection operator $\pi_0^{\Q}$ within the range of $\Q$, i.e. $\pi_0^{\Q}|_{\Q H}$, is $\pi_0|_{\Q H}$:
\begin{prop}\label{prop2}
Suppose operators $\P$ and $\K$ satisfy all conditions listed in Theorem \ref{thm:QKQ_hypocoercive}. Then $\pi_0^{\Q}|_{\Q H}=\pi_0|_{\Q H}$.
\end{prop}
The proof of this proposition is rather technical hence will be deferred to \ref{APP:proof}. We note that this result makes sense of  \eqref{eqn:e^{-tQKQ}} since the kernel projection operator $\pi_0^{\Q}$ does not appear in the semigroup estimate. In the following section, we will show that the subspace $\Q H$ is already large enough to derive exponentially decaying estimates for the memory kernel $K(t)$ and the fluctuation force $f(t)$. Moreover, the equivalence of $\pi_0^{\Q}|_{\Q H}$ and $\pi_0|_{\Q H}$ enables us to give the explicit expression of the equilibrium state for $K(t)$ and $f(t)$. 

At last, we note that Villani's hypocoercivity framework \cite{villani2009hypocoercivity} cannot be directly applied here because for the auxiliary Sobolev norm considered therein, additional terms that similar  
to $(II)$ cannot be merged into other coercivity estimates like what we did in \eqref{DgQ_estimate}, which leads to the failure of the whole scheme. 
\section{Application to the Langevin dynamics}
\label{sec:app}
In this section, we apply the above theoretical results to the Langevin dynamics and show that the EMZ memory 
kernel and fluctuation force decay exponentially fast in time to the equilibrium state. Before getting into the analysis of this specific stochastic system, we  note that since the obtained abstract hypocoercivity result for $e^{t\Q\K\Q}$ requires only the projection operator $\P$ to be symmetric and finite-rank, at least in principle, it should apply to {\em any} stochastic systems that the DMS-GS framework works. For instance, this should include the spherical Langevin dynamics considered in \cite{grothaus2014hypocoercivity}, the adaptive Langevin dynamics \cite{leimkuhler2020hypocoercivity}, and the Langevin dynamics in the smooth manifold \cite{grothaus2020hypocoercivity}.  

Now we consider the Langevin dynamics of an interactive particle system with the dynamics described by the following SDE in $\R^{2d}$:
\begin{align}\label{eqn:LE}
\begin{dcases}
dq=pdt\\
dp=-\nabla V(q)-\gamma p+\sigma d\W(t)
\end{dcases},
\end{align}
where the uniform mass $m=1$ is assumed for each particle, 
$V(q)$ is the interaction potential and 
$ \W(t)$ is a $d$-dimensional Wiener process.    
The parameters $\sigma$ and $\gamma$ represent, 
respectively, the magnitude of the fluctuations and 
the dissipation. Such parameters are linked by 
the fluctuation-dissipation relation 
$\sigma=(2\gamma/\beta)^{1/2}$, where $\beta$ 
is proportional to the inverse of the thermodynamic 
temperature. The Kolmogorov backward operator \eqref{KI} associated with 
the SDE \eqref{eqn:LE} is given by
\begin{align}\label{L:LE}
\K=p\cdot\nabla_{q}-
\nabla_{q}V( q)\cdot\nabla_{p}-
\gamma\left(p\cdot\nabla_{p}-\frac{1}{\beta}\Delta_{ p}\right),
\end{align}
which generates a Markovian semigroup $e^{t\K}$. To be noticed that, we can rewrite $\K=\D-\L$, with $\D=-\gamma p\cdot\nabla_{p}+\gamma/\beta\Delta_{p}$ and $\L=-p\cdot\nabla_{q}+\nabla_{q}V( q)\cdot\nabla_{ p}$. Here we note in particular that $\D=\X^*\X$, where $\X=\sqrt{\gamma/\beta}\nabla_p$. If the interaction potential $V( q)$ is strictly positive at infinity then the Langevin equation \eqref{eqn:LE} 
admits an unique invariant Gibbs measure given by
\begin{equation}\label{eqn_measure}
\rho_{eq}( p, q)=\frac{1}{Z} e^{-\beta \H(p,q)},
\end{equation}
where 
\begin{equation}
\H(p,q)=\frac{|p|^2}{2}+V(q) 
\end{equation}
is the Hamiltonian and $Z$ is the partition function. The equilibrium measure \eqref{eqn_measure} naturally induces a weighted Hilbert space $L^2(\R^{2d};\rho_{eq})$. Choose this as the Hilbert space $H$ studied in Section \ref{sec:ana}, we can obtain sufficient conditions to ensure the validity of $(H1)-(H4)$, therefore the semigroup estimate \eqref{e^tK_estimate}. For the Langevin dynamics \eqref{eqn:LE}, these conditions can be summarized \cite{grothaus2016hilbert} as the following requirements on the potential energy $V(q)$:  
\begin{itemize}
\item $(C1)$: $V(q)\in C^2(\R^d)$ is locally Lipshitz continuous and satisfies $\nabla_q V(q)\in H$.
\item $(C2)$: The following Poincar\'e inequality
\footnote{Without loss of generality, here we omitted $\beta$ in the probability measure $e^{-\beta V(q)}dq$. The Poincar\'e inequality \eqref{poincare_ineq} is valid for a large range of potential energy $V(q)$ which grows at least linearly as $q\rightarrow\infty$ \cite{villani2009hypocoercivity}. Specifically, this can be achieved by requiring \cite{grothaus2016hilbert,villani2009hypocoercivity}:
\begin{align*}
    \frac{|\nabla_qV(q)|^2}{2}-\Delta_q V(q)\rightarrow +\infty, \qquad |q|\rightarrow +\infty.
\end{align*}

} holds:
\begin{align}\label{poincare_ineq}
\int |\nabla_{ q}f( q)|^2e^{-V( q)}d q&\geq \Lambda\int\left[f( q)-\left(\int f(q)e^{-V( q)}d q\right)\right]^2e^{-V(q)}dq,
\end{align}
for some constant $\Lambda>0$ and for all locally Lipshitz function $f(q)$.
\item $(C3)$:  There exists a constant $C>0$ such that 
    \begin{align}
        |\nabla_q^2V(q)|\leq C(1+|\nabla_q V(q)|),\qquad \forall q\in \R^d.
    \end{align}
\end{itemize}
It is proved in \cite{grothaus2016hilbert} (Section 3.2) that if $V(q)$ satisfies $(C1)-(C3)$, then $(H1)-(H4)$ hold, which implies \eqref{e^tK_estimate}. Here we only need to note that only at $\text{\Ker}(\K)^{\perp}$, the projection operators $P$ and $P_{S}$ appearing in \cite{grothaus2016hilbert} coincide. This is one of the key points that allow the transfer of the DMS result for the Fokker-Planck equation to the GS result for the Kolmogorov backward equation. Immediately we can use Theorem \ref{thm:QKQ_hypocoercive} to get the following result on the ergodicity of $e^{t\Q\K\Q}$ for the Langevin dynamics:
\begin{prop}\label{prop:QKQ_hypocoercive}
Assume that the potential energy $V(q)$ in the Langevin dynamics \eqref{eqn:LE} satisfies the hypocoercivity conditions $(C1)-(C3)$, then for the Mori-type projection operator \eqref{Mori_P}, the EMZ orthogonal semigroup $e^{t\Q\K\Q}$ satisfies the exponentially decaying estimate:
\begin{align}\label{eqn:e^{-tQKQ}_langevin}
   \|e^{t\Q\K\Q}f_0\|\leq Ce^{-\lambda_{\Q}t}\|f_0\|, \qquad \forall f_0\in \Q H/\text{Ker}(\K).
\end{align}
\end{prop}
\subsection{Prior estimates for the EMZ equation}
Prior estimates for the memory kernel and fluctuation force term of EMZ equation are useful in applications. Recently developed numerical schemes in the coarse-grained modeling of molecular systems often pre-assume that the memory kernel decays to 0 exponentially in the long-time limit and design ansatz based on this fact \cite{Li2015, lin2021data,wang2019implicit}. Our hypocoercivity analysis for the $e^{t\Q\K\Q}$ will confirm that this is indeed the case, at least for the Mori-type EMZ equation. Without loss of generality, we only consider the estimate for scalar observable $ u(t)=u( p(t), q(t))$ and a simple Mori-type projection operator $\P=\langle\cdot, u_0\rangle u_0$. The result can be readily generalized to vector equation \eqref{gle_full} and \eqref{gle_projected} following the argument used in \cite{zhu2020hypoellipticity}. First we note that the ergodicity result for $e^{t\K}$ also implies the exponentially decaying of the stationary time autocorrelation function of the observable $u(t)$:
\begin{corollary} \label{cor:c}
We assume $(C1)-(C3)$ for the Langevin dynamics \eqref{eqn:LE}. Then the time-autocorrelation function for any observable $u(t)=u(p(t), q(t)))$ with initial condition $u_0\in H$ converges to the equilibrium state $\langle u(0)\rangle^2$ exponentially fast:
\begin{align}\label{C_est}
|\langle u(t),u(0)\rangle-\langle u(0)\rangle^2|\leq 
C e^{-\lambda t},
\end{align}
where $C=C(u(0))$. 
\end{corollary}
\begin{proof}
Using Cauchy-Schwartz inequality, the semigroup estimate \eqref{e^tK_estimate} and the fact that $\pi_0[(\cdot)]=\mathbb{E}[(\cdot)]$, it is easy to get the result.
\end{proof}
Estimate \eqref{C_est} is for the stationary time autocorrelation function since the inner product $\langle \cdot,\cdot\rangle$ is with respect to the equilibrium steady state $\rho_{eq}=e^{-\beta \H}/Z$. Using semigroup estimate \eqref{e^tK_estimate}, we can also get that in the nonequilibrium, say the system evolves from a initial state $\rho_0\neq \rho_{eq}$, the arbitrary moment of the observable $u(t)$ decays exponentially in the long-time limit. Relevant analysis is provided in \cite{zhu2021effective}. Using the semigroup estimate for the $e^{t\Q\K\Q}$, we can get the prior estimate for the memory kernel: 
\begin{corollary} \label{cor:K}
We assume $(C1)-(C3)$ for the Langevin dynamics \eqref{eqn:LE}. If $\P$ is a Mori-type projection operator corresponding to a zero-mean observable with $\langle u_0\rangle=0$, then for the EMZ equations \eqref{gle_projected}-\eqref{gle_full}, the one-dimensional memory kernel \eqref{memory_kernel} converges to the equilibrium state $\mathbb{E}[\Q\K^*u_0]\mathbb{E}[\Q\K u_0]$ exponentially fast:
\begin{align}\label{kernel_est}
|K(t)-\mathbb{E}[\Q\K^*u_0]\mathbb{E}[\Q\K u_0]|\leq 
C e^{-\lambda_{\Q} t}.
\end{align}
\end{corollary}
\begin{proof}
We first note that $\Q\K u_0-\pi_0\Q\K u_0\in \Ker(\K)^{\perp}$ since $\pi_0$ is an orthogonal projection operator. On the other hand, $\langle u_0\rangle=0$ implies $\P\pi_0\Q\K u_0=\langle\mathbb{E}[\Q\K u_0],u_0\rangle u_0=\mathbb{E}[\Q\K u_0]\langle u_0\rangle u_0=0$ and $\pi_0\Q\K u_0=\Q\pi_0\Q\K u_0$. Combining with the fact that $\Q\K u_0\in \Q H$, we obtain $\Q\K u_0-\pi_0\Q\K u_0\in \Q H\cap \Ker(\K)^{\perp}$. Now using the definition \eqref{memory_kernel} and the Cauchy-Schwartz inequality, we obtain
\begin{align*}
| K(t)-\mathbb{E}[\Q\K^*u_0]\mathbb{E}[\Q\K u_0]|
&=
| K(t)-\langle\Q\K^*u_0,\pi_0\Q\K u_0\rangle|\\
&
=|\langle u_0,\K e^{t\Q\K\Q}\Q\K u_0\rangle-\langle\Q\K^*u_0,\pi_0\Q\K u_0\rangle|\\
&=|\langle \Q\K^*u_0,e^{t\Q\K\Q}\Q\K u_0\rangle-\langle\Q\K^*u_0,\pi_0\Q\K u_0\rangle|\\
&\leq \|\Q\K^*u_0\|\|e^{t\Q\K\Q}\Q\K u_0-\pi_0\Q\K u_0\|.
\end{align*}
Since $\Q\K u_0\in \Q H$, using Proposition \ref{prop2}, we obtain $\pi_0\Q\K u_0=\pi_0^{\Q}\Q\K u_0$. Combining with the semigroup estimate \eqref{eqn:e^{-tQKQ}_langevin}, we get the final estimate:
\begin{align*}
| K(t)-\mathbb{E}[\Q\K^*u_0]\mathbb{E}[\Q\K u_0]|
&\leq \|\Q\K^*u_0\|\|e^{t\Q\K\Q}\Q\K u_0-\pi_0\Q\K u_0\|\\
&=\|\Q\K^*u_0\|\|e^{t\Q\K\Q}\Q\K u_0-\pi_0^{\Q}\Q\K u_0\|\\
&=\|\Q\K^*u_0\|\|e^{t\Q\K\Q}(\Q\K u_0-\pi_0^{\Q}\Q\K u_0)\|
\leq Ce^{-\lambda_{\Q}t},
\end{align*}
where the constant $C=C(\Q\K u_0,\Q\K^* u_0,\pi_0^{\Q}\Q\K u_0)=C(\Q\K u_0,\Q\K^* u_0,\mathbb{E}[\Q\K u_0])$. 
\end{proof}
%
%
Estimate \eqref{kernel_est} is an improvement of the result obtained in \cite{zhu2020hypoellipticity} where the equilibrium state of $K(t)$ is expressed in an abstract form $\langle\Q\K^* u_0,\pi_0^{\Q}\K u_0\rangle$. We are able to give an explicit expression $\mathbb{E}[\Q\K^*u_0]\mathbb{E}[\Q\K u_0]$ essentially because the discovery of the relation $\pi_0^{\Q}|_{\Q H}=\pi_0|_{\Q H}$\footnote{More specifically, since $\P\K u_0\in \text{Ran}(\P)\subset \text{Ker}(\Q\K\Q)$, we have $\pi_0^{\Q}\P\K u_0=\P \K u_0$. Hence $\langle \Q\K^* u_0,\pi_0^{\Q}\P\K u_0\rangle=\langle\K^* u_0,\Q\P\K u_0\rangle=0$ because $\Q\P=0$. This further implies $\langle \Q\K^* u_0,\pi_0^{\Q}\K u_0\rangle=\langle\Q\K^* u_0,\pi_0^{\Q}\Q\K u_0\rangle=\langle\Q\K^* u_0,\pi_0\Q\K u_0\rangle=\mathbb{E}[\Q\K^*u_0]\mathbb{E}[\Q\K u_0]$}.
Similarly, we have the estimate for the fluctuation force:
\begin{corollary} \label{cor:f}
We assume $(C1)-(C3)$ for the Langevin dynamics \eqref{eqn:LE}. If $\P$ is a Mori-type projection operator corresponding to a zero-mean observable with $\langle u_0\rangle=0$, then the stationary time autocorrelation function\footnote{Note that the second fluctuation-dissipation theorem, i.e. $K(t)=-\langle f(t),f(0)\rangle$ does not hold for the EMZ equation of the stochastic system since $\K$ is not a skew-symmetric operator \cite{zhu2018faber,zhu2020hypoellipticity}.} of the one-dimensional fluctuation force \eqref{f_force} satisfies:
\begin{align}\label{<f_0,f_t>}
|\langle f(t),f(0)\rangle-\mathbb{E}^2[\Q\K u_0]|\leq 
C e^{-\lambda_{\Q} t}.
\end{align}
$f(t)$ also converges to the equilibrium state $\mathbb{E}[\Q\K u_0]$ exponentially fast:
\begin{align}\label{QKQ_estimation_force}
\|f(t)-\mathbb{E}[\Q\K u_0]\|\leq Ce^{-\lambda_{\Q} t}.
\end{align}
\end{corollary} 
\begin{proof}
Using the definition \eqref{f_force} and the proof of Corollary \ref{cor:K}, it is easy to get \eqref{<f_0,f_t>} and \eqref{QKQ_estimation_force}. Note that the constant $C=C(\Q\K u_0,\pi_0^{\Q}\Q\K u_0)=C(\Q\K u_0,\mathbb{E}[\Q\K u_0])$.
\end{proof}
\subsection{Comparison with the hypoellipticity method}
It is constructive to compare the hypoellipticity method and the hypocoercivity method for the analysis of the EMZ equation.
First of all, both methods lead to the exponentially decaying estimate for the semigroup $e^{t\Q\K\Q}$, the EMZ memory kernel and the fluctuation force. The hypoellipticity method \cite{zhu2020hypoellipticity} gets these results through functional calculus and spectrum estimates. As we have seen, the hypocoercivity method is based on the classical functional analysis and the estimates are obtained by analyzing the auxiliary norm $H_{\epsilon}[\cdot]$. Roughly speaking, the result of hypocoercivity analysis is more general since the imposed conditions on the potential energy $V(q)$, i.e. $(C1)-(C3)$, are more general and easier to check (one may compare them with conditions provided in \cite{herau2004isotropic,nier2005hypoelliptic}). However, we note that the hypoellipticity analysis gives us a little bit more on the spectrum information of the operator $\Q\K\Q$ by confirming that it is discrete and confined in a cusp region in $\mathbb{C}$ \cite{zhu2020hypoellipticity}, just like the spectrum of operator $\K$ \cite{herau2004isotropic,nier2005hypoelliptic}. Lastly, we note a rather technical but important deficit/difficulty of the hypoellipticity method for analyzing the semigroup $e^{t\Q\K\Q}$. In order to obtain the exponentially decaying estimate of $e^{t\Q\K\Q}$, it is required prove that operator $\Q\K\Q$ has no imaginary eigenvalues, i.e. $\sigma(\Q\K\Q)\cap i\R=\{0\}$. This is a simple task for operator $\K$ \cite{nier2005hypoelliptic}, while non-trivial for operator $\Q\K\Q$ essentially because the kernel of $\Q\K\Q$ is multi-dimensional and strongly depends on the form of projection operator $\P$. More discussions in this regard can be found in \cite{zhu2020hypoellipticity}.
\section{Conclusions and open problems}
\label{sec:summary}
In this paper, we provide a simple but powerful hypocoercivity method to analyze the dynamical properties of the Mori-type, effective Mori-Zwanzig equation (EMZ) corresponding to degenerate stochastic systems. Our analysis shows that under the same conditions that lead to the geometric ergodicity of $e^{t\K}$, the orthogonal semigroup $e^{t\Q\K\Q}$ is also geometrically ergodic. As a result, the EMZ memory kernel and the fluctuation force decay to their equilibrium states exponentially fast, which confirm similar conclusions previously obtained by the hypoellipticity method \cite{zhu2020hypoellipticity} and pre-assumed in the reduced-order modeling \cite{Li2015,lin2021data,wang2019implicit}. The presented method is developed based on the elegant hypocoercivity framework introduced by Dolbeault, Mouhot and Schmeiser \cite{dolbeault2015hypocoercivity} and extended by Grothaus and Stilgenbauer \cite{grothaus2014hypocoercivity,grothaus2016hilbert}.
In particular, we note that the proof becomes direct once we identify a particular, invariant subspace $\Q H\cap\text{Ker}(\K)^{\perp}$ and perform hypocoercivity analysis therein. The method is expected to apply to any degenerate stochastic systems where the DMS-GS hypocoercivity framework works \cite{grothaus2014hypocoercivity,grothaus2020hypocoercivity}. At last, we propose several interesting open problems on the analysis of the EMZ equation:
\begin{enumerate}
    \item Is it possible to get $\lambda_{\Q}<\lambda$? Through the hypoellipticity and the hypocoercivity analysis, we can only get a convergence rate $\lambda_{\Q}$ for semigroup $e^{t\Q\K\Q}$ larger than $\lambda$ (Compare \eqref{DgQ_estimate} with \eqref{Dg_estimate}). However, numerical simulation results \cite{zhu2021effective,Li2015,lei2016data} for molecular systems often suggest the opposite, i.e. $\lambda_{\Q}<\lambda$, which means $\Q\K\Q$ has a larger spectrum gap and the memory kernel decays in a rate faster than the original dynamics. In fact, this is one of the main motivations to use the Mori-Zwanzig equation to do reduced-order modelling. Is it possible to prove $\lambda_{\Q}<\lambda$ rigorously for a certain degenerate stochastic system\footnote{For non-degenerate stochastic systems, this might be a simpler task. Consider the $d$-dimensional quantum harmonic oscillator \cite{Justin} with $\lambda$-coercive generator $\D=\X^*\X$, we have $\langle \Q\D\Q f,f\rangle=\langle \D\Q f,\Q f\rangle\geq \lambda_{\Q}\|\Q f\|^2=\lambda_{\Q}\|f\|^2$ for $f\in \Q H\cap\text{Ker}(\X)^{\perp}$. Since $\Q H\cap\text{Ker}(\X)^{\perp}\subset \text{Ker}(\X)^{\perp}\subset H$, we have $\lambda_{\Q}\leq \lambda$.}?
    \item What if $\P$ is an infinite-rank projection operator? Currently, we only developed the hypoellipticity and hypocoercivity methods that work for finite-rank projection operators. Another frequently used projection operator, Zwanzig-type operator \cite{zwanzig1961memory,chorin2000optimal,chorin2002optimal,hudson2018coarse}, is an infinite-rank projection operator \cite{zhu2018estimation,chorin2002optimal}. It is commonly assumed in the coarse-grained modeling that $e^{t\Q\K\Q}$ is also geometrically ergodic under the Zwanzig-type projection, but how to prove it rigorously? 
    \item What if $V(q)$ is singular? Can we get similar results for singular potential energy $V(q)$ such as the frequently used Lennard-Jones potential?
     \item Can we use the stochastic analysis framework to prove the geometric ergodicity of $e^{t\Q\K\Q}$, such as the ones used in \cite{mattingly2002ergodicity,meyn2012markov,rey2002exponential}?
\end{enumerate}
\appendix
\section{Proof of $\pi_0^{\Q}|_{\Q H}=\pi_0|_{\Q H}$}\label{APP:proof}
In this section, we prove $\pi_0^{\Q}|_{\Q H}=\pi_0|_{\Q H}$. In fact, this is a corollary of Lemma 5 proved in \cite{zhu2020hypoellipticity}. In the context of our paper, the result of Lemma 5 \cite{zhu2020hypoellipticity} can be rephrased as:
\begin{lemma}(Zhu and Venturi \cite{zhu2020hypoellipticity}) For SDE \eqref{eqn:sde}, if $\P: H\rightarrow H$ is a Mori-type projection operator \eqref{Mori_P}, then we have:
\begin{align}
\textrm{Ker}(\Q\K\Q)
&\subset\text{Ran}(\P)\cup
\text{Ker}(\K)\cup\text{Span}\{w_j\}_{j=1}^N,\label{1}
\end{align}
where $w_j$ satisfies $\Q w_j=w_j$, $\K w_j=u_j$ and $\langle w_j,u_i\rangle=0$ for all $1\leq i,j\leq N$. Here $u_j$ is the observable function defined in \eqref{Mori_P}.  
\end{lemma}
\paragraph{Proof of Proposition \ref{prop2}}
Since $\pi_0^{\Q}$ is the kernel projection operator for $\Q\K\Q$, we have $\Ran(\pi_0^{\Q})=\textrm{Ker}(\Q\K\Q)$. This implies
\begin{equation}\label{2}
\begin{aligned}
\Ran(\pi_0^{\Q}|_{\Q H})=\Q H\cap \textrm{Ker}(\Q\K\Q)
&\subset
\Q H\cap
\left(\text{Ran}(\P)\cup
\text{Ker}(\K)\cup\text{Span}\{w_j\}_{j=1}^N\right)\\
&\subset\left(\Q H\cap\text{Ker}(\K)\right)\cup\left(\Q H\cap\text{Span}\{w_j\}_{j=1}^N\right)\\
&=\left(\Q H\cap\text{Ker}(\K)\right)\cup\text{Span}\{w_j\}_{j=1}^N
\end{aligned}
\end{equation}
Here we have used the set distributive rule, the fact that $\Q H\cap \Ran(\P)=0$ and $\Q w_j=w_j$. On the other hand, since $\langle w_j, u_j\rangle=\langle w_j, \K w_j\rangle=\| \X w_j\|^2=0$, we have $\K^* w_j=2\X^*\X w_j-\K w_j=-\K w_j$. In the proof of Theorem \ref{thm:QKQ_hypocoercive}, we have shown that $\Q H\cap\Ker(\K)^{\perp}$ is an invariant subspace of operator $\Q\K\Q$. Naturally we have $\Q H\cap\Ker(\K)^{\perp}\subset \Ran(\Q\K\Q)$. This implies for any $f\in \Q H\cap\Ker(\K)^{\perp}$, there exists a $g\in H$ such that $\Q\K\Q g=f$. Then we obtain
\begin{align*}
    \langle f,w_j\rangle=\langle \Q\K\Q g,w_j\rangle
    =\langle \K\Q g,w_j\rangle
    =\langle \Q g,\K ^*w_j\rangle
    =-\langle \Q g,\K w_j\rangle
    =-\langle g,\Q\K w_j\rangle.
\end{align*}
Note that $\K w_j=u_j$ and $\P u_j=u_j$, we have $\Q\K w_j=0$ and $\langle f,w_j\rangle=0$. Combining with the fact that $w_j\in \Q H$ and the equality $\langle f,w_j\rangle=0$ holds for any $f\in \Q H\cap\Ker(\K)^{\perp}$, we know $w_j\in \Q H\cap\Ker(\K)$. Here we have used the De Morgan's law: $(A\cap B)^{\perp}=A^{\perp}\cup B^{\perp}$. \eqref{2} can be further simplified as 
\begin{equation}\label{3}
\Ran(\pi_0^{\Q}|_{\Q H})\subset\left(\Q H\cap\text{Ker}(\K)\right)\cup\text{Span}\{w_j\}_{j=1}^N
=\Q H\cap\Ker(\K)
\end{equation}
On the other hand, if $f\in\Q H\cap\Ker(\K)$, we have $\Q\K\Q f=\Q\K f=0$. Hence $f\in \Ker(\Q\K\Q)\cap\Q H=\Ran(\pi_0^{\Q}|_{\Q H})$ and $\Q H\cap\Ker(\K)\subset \Ran(\pi_0^{\Q}|_{\Q H})$. Combining with \eqref{3}, we know $\Ran(\pi_0^{\Q}|_{\Q H})=\Q H\cap \Ker(\K)$. Since the orthogonal projection operator on $\Q H\cap \Ker(\K)$ is given by $\pi_0|_{\Q H}$, we get  $\pi_0^{\Q}|_{\Q H}=\pi_0|_{\Q H}$.

\vspace{0.3cm}
\noindent 
\bibliographystyle{plain}
\bibliography{Hypocoercivity_EMZ}
\end{document}